\documentclass{article}
\usepackage[main, preprint]{neurips_2026}
\usepackage{float}
\usepackage[utf8]{inputenc}
\usepackage[T1]{fontenc}
\usepackage{amsmath,amssymb,amsthm}
\usepackage{mathtools}
\usepackage{bm}
\usepackage{graphicx}
\usepackage{booktabs}
\usepackage{hyperref}
\usepackage{url}
\usepackage{xcolor}

\newcommand{\indep}{\mathrel{\perp\!\!\!\perp}}

\newtheorem{lemma}{Lemma}[section]
\newtheorem{theorem}{Theorem}[section]

\newtheorem{proposition}{Proposition}[section]

\setcitestyle{semicolon}
\title{FlowSDR: Sufficient Dimension Reduction via Conditional Normalizing Flows}

\author{%
    Yuexiao Dong\\
  Temple University \\
  Philadelphia, PA 19102\\
 \And
    Kenichiro Mcalinn \\
  Temple University \\
  Philadelphia, PA 19102\\
  \And
   Edoardo  Airoldi \\
  Temple University \\
  Philadelphia, PA 19102\\
  \And
    Lei Li \\
  LunarAI LLC \\
  Newark, DE 19702 \\
  }

\begin{document}
\maketitle

\begin{abstract}
Sufficient dimension reduction (SDR) seeks a low-dimensional linear projection $B^\top X$
 of predictors that preserves the conditional distribution of the response $Y$.
  Existing methods target this conditional distribution indirectly, via inverse moments, local forward regression, or neural ensemble regression. We propose FlowSDR, a likelihood-based framework that jointly learns the projection $B^\top X$ and the conditional density $f(Y |B^\top X)$  by maximizing a conditional log-likelihood, with the density parameterized by monotone rational-quadratic spline flows. The estimator is Fisher consistent under the SDR model, and its sample  objective admits a population interpretation in terms of mutual information. As a complementary model within the same likelihood framework, we introduce the neural Gaussian SDR, a heteroscedastic conditional Gaussian model whose mean and variance are parameterized by
shared neural-network functions of $B^\top X$. 
In simulations spanning Gaussian errors, heavy-tailed distributions, two-component mixtures, and settings with tail behavior not captured by mean–variance structure, FlowSDR recovers the central subspace more accurately than existing SDR methods and the neural Gaussian SDR baseline. We further validate these advantages on a face-age prediction task using the UTKFace dataset.
\end{abstract}

\section{Introduction and related work}

High-dimensional predictors are ubiquitous in modern applications, yet the information relevant to predicting a response often lies in a low-dimensional subspace. Sufficient dimension reduction (SDR) \citep{li1991sir,cook2007dimension,li2018sdr,cook2026foundations} provides a principled framework for identifying such structure by seeking a low-dimensional projection of the predictors that preserves the conditional distribution of the response. Formally,
given a predictor \(X \in \mathbb{R}^p\) and a  response \(Y \in \mathbb{R}\),
 SDR aims to find a matrix  $B \in \mathbb{R}^{p \times d}$ with  $d < p$, such that
\begin{equation}
\label{eq:sdr}
  Y \indep X \mid B^{\top}X,
\end{equation}
where $\indep$ means statistical independence. Model (\ref{eq:sdr}) implies that 
the projected feature \(B^{\top}X \in \mathbb{R}^d\) is sufficient for modeling the conditional distribution of \(Y\). The subspace spanned by the columns of $B$ satisfying (\ref{eq:sdr}) is called a dimension
reduction subspace. Among all such
subspaces, their intersection defines the central subspace, which is
the minimal subspace that preserves the conditional distribution of
$Y$ given $X$.
Sufficient dimension reduction has played critical roles in various statistical and machine learning problems, such as  conditional density estimation \citep{tangkaratt2014conditional}, function approximation \citep{zhang2019learning}, classification \citep{meng2020sufficient,zhang2020maximum},  survival analysis \citep{ding2020double}, online learning 
\citep{cai2020online,artemiou2021realtime}, 
interaction detection \citep{tang2020highdim}, and causal inference \citep{luo2020matching,luo2022efficient}. These applications highlight the importance of identifying low-dimensional structures that retain predictive information while mitigating the curse of dimensionality.

Classical sufficient dimension reduction methods, including sliced inverse regression (SIR)
\citep{li1991sir}
and sliced average variance estimation (SAVE) 
\citep{cook1991save}, estimate the low-dimensional central subspace through moment-based or inverse-regression techniques.
By combining local linear regression with predictor dimension reduction, minimum average variance estimation (MAVE) \citep{xia2002mave} and its variation central subspace MAVE (csMAVE) \citep{wang2008sliced} achieve sufficient dimension reduction via forward regression.
Despite their success, traditional SDR methods face several limitations. Inverse-regression approaches such as SIR and SAVE rely on restrictive  assumptions about the predictor distribution, which may not hold in applications. Forward-regression methods like MAVE and csMAVE improve flexibility but can be sensitive to bandwidth selection. 
This line of work has been extended through machine learning techniques, including local SDR \citep{wu2008localized}, principal support vector machines \citep{li2011psvm}, RKHS-based methods \citep{fukumizu2004dimensionality,fukumizu2009kernel,hsing2009rkhs}, and nonlinear SDR \citep{yeh2008nonlinear,li2017nonlinear}. More recently, neural-network-based SDR methods have been proposed to further enhance modeling flexibility \citep{kapla2022fusing,liang2022nonlinear,chen2024deep,yang2025golden,tang2026belted}.

Despite the increased flexibility offered by these approaches, most existing SDR methods focus primarily on estimating the central subspace, rather than directly modeling the conditional distribution \(f_{Y | X}(y | x)\). As a result, they may fail to fully capture complex relationships between predictors and response, particularly in settings involving heavy-tailed errors, heteroscedasticity, or multimodal conditional structures. 

To address these limitations, it is desirable to develop an SDR framework that directly targets the conditional distribution while preserving the dimension reduction structure. In this paper, we propose a novel SDR approach based on normalizing flows. Normalizing flows \citep{papamakarios2017maf,dinh2017realnvp,papamakarios2021normalizing} provide a flexible yet tractable class of invertible transformations that map complex distributions to simple base distributions. As a result, they can approximate conditional distributions \(f_{Y | X}(y | x)\) with high fidelity, capturing non-Gaussian features such as skewness, multimodality, and heteroscedasticity. Our contributions are summarized as follows:

\begin{enumerate}
\setlength{\itemsep}{0.4em}
\setlength{\parskip}{0pt}
\setlength{\parsep}{0pt}
\setlength{\leftmargin}{0pt}
\item[1.]
We propose FlowSDR, a likelihood-based sufficient dimension reduction method that jointly learns a low-dimensional representation $B^\top X$ and the conditional density $f(Y | B^\top X)$ using conditional normalizing flows. We establish a population-level interpretation of the proposed objective through an information index closely related to mutual information, show that the resulting population target is Fisher consistent under the SDR model, and prove consistency of the sample-level estimator.

\item[2.]
As a complementary likelihood-based baseline, we introduce neural Gaussian SDR, a heteroscedastic conditional Gaussian model that clarifies the role of flexible conditional density estimation within the SDR framework. We further show that the population-level target of neural Gaussian SDR is closely connected to MAVE,  a prominent classical SDR estimator. 

\item[3.]
Through simulations involving heteroscedasticity, heavy-tailed errors, multimodal responses, and tail-dependent conditional structures, we demonstrate that FlowSDR substantially improves central subspace recovery compared with existing SDR methods.

\item[4.]
We further validate the proposed framework on the UTKFace age prediction task, where FlowSDR achieves strong predictive and conditional density estimation performance while naturally enabling conformal prediction intervals.
\end{enumerate}

\section{FlowSDR}
\label{sec:flow sdr}

\subsection{Conditional flow model}

Let \(W = B^{\top}X \in \mathbb{R}^d\) denote the projected feature satisfying
 $ Y \indep X \mid W$. We model the 
 
 \noindent conditional distribution of \(Y \mid W=w\) using a conditional normalizing flow.

Let \(T_{\theta}(\,\cdot\,; w): \mathbb{R} \to \mathbb{R}\) be an invertible transformation indexed by the conditioning value \(W=w\). Define
\[
  Z = T_{\theta}(Y; W).
\]
The model assumes that, for each \(w\), the transformed response $Z$ has a fixed base density \(f_0\), typically the standard normal density:
\[
  Z \mid W=w \sim f_0 .
\]
Equivalently, all dependence of the conditional distribution of \(Y\) on \(w\) is represented through the transformation \(T_{\theta}(\,\cdot\,; w)\), while the base distribution itself is free of \(w\).
By the change-of-variables formula, the induced conditional density of \(Y\) given \(W=w\) is
\begin{equation}
\label{eq:flow model}
  f_{\theta}(y | w)
  =
  f_0\{T_{\theta}(y; w)\}
  \left|
  \frac{\partial T_{\theta}(y; w)}{\partial y}
  \right|.
\end{equation}

We parameterize $T_{\theta}(\,\cdot\,; w)$ using monotone rational-quadratic splines (RQS) \citep{durkan2019neural}. For each $w$, the transformation is defined by a piecewise rational-quadratic function whose parameters, which include bin widths, heights, and derivatives, are produced by a neural network. 

Specifically,  define knot sequences for fixed $w$ and fixed number of bins $K$ as
\[
y_0(w) < y_1(w) < \cdots < y_K(w),
\mbox{ and } 
z_0(w) < z_1(w) < \cdots < z_K(w).
\]
Let $s_k(w) > 0$ 
denote the derivatives at the knots. All quantities
$\{y_k(w), z_k(w), s_k(w)\}_{k=0}^K$ are
determined by a neural network $\phi_\theta(w)$,
typically by outputting unconstrained parameters that are transformed to
satisfy positivity and ordering constraints.

Let $\Delta y_k(w) = y_{k+1}(w) - y_k(w) $ and $\Delta z_k(w) = z_{k+1}(w) - z_k(w) $. 
For \(y \in [y_k(w), y_{k+1}(w)]\), define 
\[
r_k(w) = \frac{\Delta z_k(w)}{\Delta y_k(w)} \mbox{ and }  \xi(y,w) = \frac{y - y_k(w)}{\Delta y_k(w)} \in [0,1].
\]
Then the transformation is
\[
T_{\theta}(y; w)
=
z_k(w)
+
\Delta z_k(w)
\frac{
r_k(w)\,\xi(y,w)^2 + s_k(w)\,\xi(y,w)\big(1 - \xi(y,w)\big)
}{
r_k(w) + \big(s_{k+1}(w) + s_k(w) - 2r_k(w)\big)\,\xi(y,w)\big(1 - \xi(y,w)\big)
}.
\]
Under the constraints
$\Delta y_k(w) > 0$,  $\Delta z_k(w) > 0$, and $s_k(w) > 0$,
the transformation \(T_{\theta}(y; w)\) is strictly monotone in \(y\)
for each fixed \(w\), continuously differentiable, invertible, and admits
a closed-form derivative.
  This yields a flexible family of smooth, invertible transformations with tractable Jacobians.

  While general-purpose normalizing flows such as NICE \citep{dinh2014nice}, RealNVP \citep{dinh2017realnvp}, and MAF \citep{papamakarios2017maf} are widely used for multivariate density estimation, they are less well-suited for modeling univariate conditional distributions. In one dimension, coupling-based flows such as NICE and RealNVP reduce to affine transformations, while autoregressive flows (MAF) parameterize location–scale transformations of a base distribution, thereby limiting the expressiveness of the transformation in the response variable. In contrast, monotone spline-based flows directly parameterize a rich class of nonlinear  transformations, making them particularly effective for conditional density estimation in one dimension.

\subsection{Central subspace estimation}

Given an i.i.d. sample \(\{(X_i, Y_i)\}_{i=1}^n\), the log-likelihood  under the conditional flow model (\ref{eq:flow model}) can be written as 
\begin{equation*}
\begin{aligned}
\ell_n(B,\theta)
&=
\sum_{i=1}^n
\log f_\theta(Y_i | B^\top X_i)\\
&=
\sum_{i=1}^n
\left[
\log f_0\{T_\theta(Y_i;B^\top X_i)\}
+
\log \left|
\frac{\partial T_\theta(y;B^\top X_i)}{\partial y}
\Bigg|_{y=Y_i}
\right|
\right].
\end{aligned}
\end{equation*}
Let $\Theta$ denote the parameter space indexing the conditional normalizing flow family
$\{ f_\theta(Y | B^\top X) : \theta \in \Theta \}$.
For a fixed matrix \(B\), define the empirical profile log-likelihood
\begin{equation}
\label{eq:sample profile}
M_n(B)
=
\sup_{\theta\in\Theta}
\frac{1}{n}
\ell_n(B,\theta)
=
\sup_{\theta\in\Theta}\frac{1}{n}\sum_{i=1}^n\log f_\theta(Y_i| B^\top X_i).
\end{equation}
Our proposed method estimates the sufficient reduction subspace by solving
\begin{equation}
\label{eq:sample}
\hat{B}
=
\arg\max_{ B\in \mathbb{R}^{p \times d}}\
M_n(B) \quad \text{subject to } B^\top B = I_d.
\end{equation}
The constraint \(B^\top B = I_d\) in (\ref{eq:sample}) restricts \(B\) to lie on the Stiefel manifold, and the optimization is carried out using gradient-based methods on this manifold \citep{edelman1998geometry}. The column space of $\hat B$, or $\mathrm{span}(\hat B)$, is taken as the final FlowSDR estimator of the central subspace.

\subsection{Population-level targets}

 For notational convenience, we use $f(Y | X)$ for the conditional density of $Y$ given $X$, and $f(Y)$ as the marginal density of $Y$. The symbol $f$ is used generically to denote densities in each context.  For $A\in \mathbb{R}^{p \times q}$ with $q\le p$,
denote $f(Y | A^\top X)$ as the conditional density of $Y$ given $A^\top X$.
Consider the following information index $\mathcal{J}(A)$ 
\begin{equation}
\label{eq:mi}
   \mathcal{J}(A)={\rm E}\left\{ \log{f(Y | A^\top X)} \right\}.
\end{equation}
Since 
$\mathcal{J}(A)-{\rm E}\{\log f(Y)\}$
coincides with the mutual information between \(Y\) and \(A^\top X\), and \({\rm E}\{\log f(Y)\}\) does not depend on \(A\),
 maximizing \(\mathcal{J}(A)\) is equivalent to maximizing the mutual information between \(Y\) and \(A^\top X\). Moreover, \(\mathcal{J}(A)\) depends only on the subspace spanned
by the columns of \(A\). In particular, for any orthogonal matrix
\(C\in \mathbb{R}^{q \times q}\) satisfying \(C^\top C=I_q\), we have
\(\mathcal{J}(A)=\mathcal{J}(AC)\).
 
 Recent works such as  \cite{ nguyen2010estimating,belghazi2018mine, oord2018cpc}
propose variational lower bounds for mutual information that avoid explicit
density estimation by learning a critic function. In contrast, 
our approach directly models the conditional density $f(Y| B^\top X)$  under model (\ref{eq:sdr}) via conditional normalizing flows, enabling likelihood-based estimation of the central subspace.

At the population level, the target parameter of FlowSDR is defined as
\begin{equation}
\label{eq:population}
B_0
=
\arg\max_{B \in \mathbb{R}^{p \times d}} \mathcal{J}(B) \quad \text{subject to } B^\top B = I_d.
\end{equation}
As shown in Appendix \ref{sec:fisher info}, under model~(\ref{eq:sdr}), the above objective is Fisher consistent in the sense that the column space of any maximizer $B_0$ coincides with the true central subspace.

For a fixed projection matrix \(B \in \mathbb{R}^{p \times d}\), we model the conditional density $f(Y| B^\top X)$ using a conditional normalizing flow family $ \left\{
  f_\theta(Y | B^\top X): \theta \in \Theta
  \right\}$.
The corresponding population profile log-likelihood for a fixed \(B\) is
\begin{equation}
\label{eq:population profile}
  M(B)
  =
  \sup_{\theta \in \Theta}
  {\rm E}\left\{
  \log f_\theta(Y | B^\top X)
  \right\}.
\end{equation}
The criterion $M(B)$ can be viewed as a model-based approximation to $\mathcal{J}(B)$ in (\ref{eq:population}), where the conditional density is restricted to the flow family.
This defines the  intermediate target parameter 
\begin{equation}
\label{eq:population i}
B^\dagger
=
\arg\max_{B \in \mathbb{R}^{p \times d}} M(B) \quad \text{subject to } B^\top B = I_d.
\end{equation}
We leave the detailed sample-level asymptotic analysis of $\hat B$  to Appendix \ref{sec:flowsdr consistency}. The estimator $\hat B$ in \eqref{eq:sample} is the sample analogue of the
population profile target $B^\dagger$ in \eqref{eq:population i}. Meanwhile, the
population target $B_0$ in \eqref{eq:population} is Fisher consistent under the
SDR model \eqref{eq:sdr}. We next clarify the connection between $B^\dagger$ and
$B_0$.

Suppose the conditional flow family is sufficiently rich in the sense that, for
each fixed $B$,
\begin{equation}
\label{eq:rich class1}
\sup_{\theta \in \Theta}
  {\rm E}\left\{
\log f_\theta(Y | B^\top X)
\right\}
\approx
  {\rm E}\left\{
\log f(Y | B^\top X)
\right\}.
\end{equation}
Then the profile likelihood criterion $M(B)$ approximates the information index $\mathcal{J}(B)$.
Consequently, the maximizers are approximately equivalent up to subspace
identifiability, namely
$\operatorname{span}(B^\dagger) \approx \operatorname{span}(B_0)$.

\section{Neural Gaussian SDR}
\label{sec:ng sdr}

As a secondary contribution, we consider the following likelihood-based SDR model 
\begin{equation}
\label{eq:ng model}
Y \mid B^\top X = w
\sim
\mathcal{N}\big( \mu_\gamma(w), \sigma_\gamma^2(w) \big),
\end{equation}
where \(\mu_\gamma(\cdot)\) and \(\sigma_\gamma^2(\cdot)\) are flexible functions jointly parameterized by a shared heteroscedastic neural network \(\psi_\gamma(w)\). See, for example, \cite{stirn2023faithful}, where a neural heteroscedastic regression model is used to characterize the relationship between $Y$ and $X$ without the linear projection $B^\top X$.

Given an i.i.d. sample \(\{(X_i, Y_i)\}_{i=1}^n\), the log-likelihood under model (\ref{eq:ng model}) is
\[
\ell_{n}^{\rm G}(B,\gamma)
=
\sum_{i=1}^n
\left\{
-\frac{1}{2}\log\big(2\pi\sigma_\gamma^2(B^\top X_i)\big)
-
\frac{\big(Y_i - \mu_\gamma(B^\top X_i)\big)^2}
{2\sigma_\gamma^2(B^\top X_i)}
\right\}.
\]

The empirical profile log-likelihood becomes
\[
M_n^{\rm G}(B)
=
\sup_{\gamma\in\Gamma}
\frac{1}{n}
\ell_n^{\rm G}(B,\gamma),
\]
where $\Gamma$ denotes the parameter space of neural network \(\psi_\gamma(w)\).
The neural Gaussian SDR estimator is obtained by solving
\[
\hat{B}^{\rm G}
=
\arg\max_{ B\in \mathbb{R}^{p \times d}}\
M_n^{\rm G}(B) \quad \text{subject to } B^\top B = I_d.
\]

As in FlowSDR, the orthogonality constraint restricts \(B\) to the Stiefel manifold, and the estimated central subspace is given by \(\operatorname{span}(\hat{B}^{\rm G})\).
This formulation provides a useful conceptual bridge between FlowSDR and conditional normalizing flows. In FlowSDR, we learn a transformation that maps the conditional distribution \(Y | B^\top X\) to a simple base distribution, typically standard normal. When the true conditional distribution is already Gaussian, this 
transformation reduces to a location-scale standardization,
$$Z = \frac{Y-\mu_\gamma(B^\top X)}{\sigma_\gamma(B^\top X)}.$$ In this sense, the neural Gaussian SDR model can be viewed as a  special case of the conditional flow framework, where the expressiveness is concentrated entirely in modeling the conditional mean and variance of the Gaussian  distribution.

From an information-theoretic perspective, the objective remains aligned with the population target
\[
\mathcal{J}(B)={\rm E}\bigl\{\log f(Y | B^\top X)\bigr\}.
\]
The key distinction is that, instead of optimizing over a rich nonparametric class induced by invertible transformations, we restrict the conditional density 
to the Gaussian family with flexible mean and variance functions. Consequently, the neural Gaussian SDR estimator can be interpreted as a constrained version of FlowSDR, where the conditional likelihood is approximated within the heteroscedastic normal class.

\section{Numerical experiments}

\subsection{Methods for empirical comparison}

In \cite{cook2026foundations}, it is noted that there have been 17 articles and four books providing substantial reviews of SDR since \cite{burges2010dimension}. As it is infeasible to compare with all existing SDR methods, we focus on a representative 
set of widely used approaches, including MAVE \citep{xia2002mave}, csMAVE \citep{wang2008sliced}, 
neural network SDR (NN-SDR)
\citep{kapla2022fusing}, and belted and ensembled neural network SDR (BENN-SDR) \citep{tang2026belted}. MAVE and csMAVE are among the most prominent forward-regression SDR methods, but they can be computationally intensive and less effective when the predictor 
dimension $p$ is large.  Unlike classical inverse-regression methods, such as
SIR
\citep{li1991sir}
and SAVE
\citep{cook1991save}, which rely on assumptions about the 
predictor distribution, NN-SDR and BENN-SDR impose fewer distributional restrictions and can flexibly model 
complex relationships, often leading to improved estimation accuracy in practice.
More detailed population-level comparisons between these benchmark methods and our proposed approaches are provided in  Appendix \ref{sec:population comparison}. We note that principal fitted components (PFC) \citep{cook2007dimension} and likelihood acquired directions (LAD) \citep{cook2009likelihood} are established likelihood-based SDR methods. These approaches are typically developed under structured parametric modeling frameworks formulated through inverse regression. 
In contrast, our method directly models the forward conditional distribution $Y|X$ using a flexible nonparametric family. Given this difference in modeling focus, we do not include them in our empirical comparisons.

\subsection{Simulation}

Let $\mathcal V_{p,d}
=
\{ B \in \mathbb{R}^{p\times d} : B^T B=I_d\}$ be the Stiefel manifold, and let
$B_0\in\mathcal V_{p,d}$ be an orthonormal basis of the true central subspace. For an estimator $\hat B\in\mathcal V_{p,d}$, define
$$\Delta_{\rm F}(\hat B)=\|\hat B \hat B^\top-B_0 B_0^\top\|_{\rm F} \mbox{ and }\Delta_{\rm op}(\hat B)=\|\hat B \hat B^\top-B_0 B_0^\top\|_{\rm op},$$
where $\|\,\cdot\,\|_{\rm F}$ and $\|\,\cdot\,\|_{\rm op}$ denote the Frobenius norm and the operator norm, respectively. 
Based on $100$ replications, we report the mean, median and  standard deviation of $\Delta_{\rm F}$ and $\Delta_{\rm op}$ for six different estimators:   NN-SDR, BENN-SDR,  MAVE, csMAVE, neural Gaussian SDR, and FlowSDR. The implementation details of all the estimators are provided in Appendix \ref{sec:details}. Throughout the simulations, we treat the structural dimension $d$ of the central subspace as given, and estimation of $d$ is discussed in Appendix \ref{sec:order}.

In Model I, we follow the linear SDR setting of    \cite{tang2026belted}. Let $X_1,\ldots, X_p \stackrel{\mathrm{i.i.d.}}{\sim} \mathcal{N}(2,1)$, $\epsilon\sim \mathcal{N}(0,1)$ is independent of $\{X_j\}_{j=1}^p$, and 
\[
  Y = 2\sin\{(X_1+X_2)\pi/20\}
      +3\sin^2\{(X_3+3X_4)\pi/30\}\epsilon.
  \qquad .
\]
The central subspace is spanned by $\beta_1=(1,1,0,\ldots,0)^\top$ and $\beta_2=(0,0,1,3,0,\ldots,0)^\top$.

\begin{table}[t]
  \centering
  \scriptsize
  \setlength{\tabcolsep}{3.5pt}
  \caption{Results for Model I with $n=2000$ and varying ambient dimension $p$. The mean, median, and standard deviation of $\Delta_{\rm op}$ and $\Delta_{\rm F}$ over $100$ replications are reported.}
  \label{tab:benn-fv-comparison}
  \begin{tabular}{clcccc}
    \toprule
    $p$ & Method & $\Delta_{\rm op}$ mean $\pm$ SD & $\Delta_{\rm op}$ median & $\Delta_{\rm F}$ mean $\pm$ SD & $\Delta_{\rm F}$ median \\
    \midrule
    5 & NN-SDR & $0.5172 \pm 0.0631$ & $0.5157$ & $0.7800 \pm 0.1015$ & $0.7804$ \\
     & BENN-SDR & $0.3859 \pm 0.0974$ & $0.3735$ & $0.5657 \pm 0.1365$ & $0.5558$ \\
     & MAVE & $0.8136 \pm 0.1875$ & $0.8519$ & $1.1838 \pm 0.2643$ & $1.2534$ \\
     & csMAVE & $0.1420 \pm 0.0509$ & $0.1375$ & $0.2124 \pm 0.0711$ & $0.2050$ \\
     & neural Gaussian SDR & $0.0481 \pm 0.0501$ & $\mathbf{0.0358}$ & $0.0729 \pm 0.0710$ & $\mathbf{0.0570}$ \\
     & \textbf{FlowSDR} & $\mathbf{0.0393 \pm 0.0144}$ & $0.0392$ & $\mathbf{0.0601 \pm 0.0211}$ & $0.0585$ \\
       \hline
    \addlinespace[2pt]
    10 & NN-SDR & $0.9698 \pm 0.0133$ & $0.9714$ & $1.4216 \pm 0.0273$ & $1.4190$ \\
     & BENN-SDR & $0.9827 \pm 0.0196$ & $0.9897$ & $1.4250 \pm 0.0406$ & $1.4272$ \\
     & MAVE & $0.9052 \pm 0.1107$ & $0.9484$ & $1.3752 \pm 0.1570$ & $1.4203$ \\
     & csMAVE & $0.2500 \pm 0.0666$ & $0.2446$ & $0.3782 \pm 0.0908$ & $0.3636$ \\
     & neural Gaussian SDR & $0.0853 \pm 0.0599$ & $0.0672$ & $0.1347 \pm 0.0851$ & $0.1124$ \\
     & \textbf{FlowSDR} & $\mathbf{0.0638 \pm 0.0139}$ & $\mathbf{0.0621}$ & $\mathbf{0.1055 \pm 0.0222}$ & $\mathbf{0.1035}$ \\
     \hline
    \addlinespace[2pt]
    20 & NN-SDR & $0.9598 \pm 0.0106$ & $0.9600$ & $1.4844 \pm 0.0390$ & $1.4812$ \\
     & BENN-SDR & $0.8149 \pm 0.1284$ & $0.8317$ & $1.2671 \pm 0.1745$ & $1.2874$ \\
     & MAVE & $0.9200 \pm 0.0854$ & $0.9492$ & $1.5159 \pm 0.1338$ & $1.5395$ \\
     & csMAVE & $0.4565 \pm 0.1382$ & $0.4217$ & $0.6856 \pm 0.1906$ & $0.6350$ \\
     & neural Gaussian SDR & $0.1504 \pm 0.0861$ & $0.1174$ & $0.2436 \pm 0.1184$ & $0.1976$ \\
     & \textbf{FlowSDR} & $\mathbf{0.1148 \pm 0.0256}$ & $\mathbf{0.1093}$ & $\mathbf{0.1969 \pm 0.0374}$ & $\mathbf{0.1926}$ \\
    \bottomrule
  \end{tabular}
\end{table}

From Table~\ref{tab:benn-fv-comparison},  FlowSDR achieves the best overall performance across all settings, with neural Gaussian SDR consistently performing as a close second. MAVE performs poorly because it fails to recover $\beta_2$, which appears only through  the conditional variance. By design, csMAVE addresses this limitation and substantially improves over the original MAVE. NN-SDR and BENN-SDR are both outperformed by csMAVE as well as our new proposals. 

In Model II, we keep all other settings the same as in Model I, except that the normal error is replaced by $\epsilon\sim t_3/\sqrt{3}$.

\begin{table}[t]
  \centering
  \scriptsize
  \setlength{\tabcolsep}{3.5pt}
  \caption{Results for Model II with $n=2000$ and varying ambient dimension $p$. The mean, median, and standard deviation of $\Delta_{\rm op}$ and $\Delta_{\rm F}$ over $100$ replications are reported.}
  \label{tab:fv-t3-comparison}
  \begin{tabular}{clcccc}
    \toprule
    $p$ & Method & $\Delta_{\rm op}$ mean $\pm$ SD & $\Delta_{\rm op}$ median & $\Delta_{\rm F}$  mean $\pm$ SD & $\Delta_{\rm F}$  median \\
    \midrule
    5 & NN-SDR & $0.5281 \pm 0.0647$ & $0.5296$ & $0.7913 \pm 0.1038$ & $0.7806$ \\
     & BENN-SDR & $0.4911 \pm 0.1262$ & $0.4929$ & $0.7422 \pm 0.1934$ & $0.7343$ \\
     & MAVE & $0.8122 \pm 0.1824$ & $0.8728$ & $1.1821 \pm 0.2649$ & $1.2552$ \\
     & csMAVE & $0.1165 \pm 0.0436$ & $0.1160$ & $0.1778 \pm 0.0628$ & $0.1782$ \\
     & neural Gaussian SDR & $0.1175 \pm 0.1255$ & $0.0782$ & $0.1752 \pm 0.1835$ & $0.1155$ \\
     & \textbf{FlowSDR} & $\mathbf{0.0469 \pm 0.0176}$ & $\mathbf{0.0474}$ & $\mathbf{0.0704 \pm 0.0252}$ & $\mathbf{0.0713}$ \\
     \hline
    \addlinespace[2pt]
    10 & NN-SDR & $0.9713 \pm 0.0120$ & $0.9732$ & $1.4281 \pm 0.0304$ & $1.4222$ \\
     & BENN-SDR & $0.9524 \pm 0.0569$ & $0.9697$ & $1.4460 \pm 0.1174$ & $1.4403$ \\
     & MAVE & $0.8964 \pm 0.1221$ & $0.9529$ & $1.3567 \pm 0.1715$ & $1.4182$ \\
     & csMAVE & $0.1819 \pm 0.0389$ & $0.1792$ & $0.2964 \pm 0.0534$ & $0.2943$ \\
     & neural Gaussian SDR & $0.2079 \pm 0.2070$ & $0.1373$ & $0.3213 \pm 0.3068$ & $0.2241$ \\
     & \textbf{FlowSDR} & $\mathbf{0.0766 \pm 0.0165}$ & $\mathbf{0.0760}$ & $\mathbf{0.1201 \pm 0.0252}$ & $\mathbf{0.1163}$ \\
     \hline
    \addlinespace[2pt]
    20 & NN-SDR & $0.9591 \pm 0.0156$ & $0.9592$ & $1.4943 \pm 0.0707$ & $1.4826$ \\
     & BENN-SDR & $0.9162 \pm 0.0825$ & $0.9373$ & $1.5336 \pm 0.1663$ & $1.5274$ \\
     & MAVE & $0.9319 \pm 0.0714$ & $0.9608$ & $1.4927 \pm 0.1187$ & $1.5104$ \\
     & csMAVE & $0.2626 \pm 0.0352$ & $0.2660$ & $0.4445 \pm 0.0479$ & $0.4394$ \\
     & neural Gaussian SDR & $0.3143 \pm 0.2382$ & $0.2178$ & $0.5128 \pm 0.3870$ & $0.3757$ \\
     & \textbf{FlowSDR} & $\mathbf{0.1190 \pm 0.0246}$ & $\mathbf{0.1161}$ & $\mathbf{0.1971 \pm 0.0406}$ & $\mathbf{0.1918}$ \\
    \bottomrule
  \end{tabular}
\end{table}

From Table~\ref{tab:fv-t3-comparison}, FlowSDR again achieves the best overall performance across all settings. 
Neural Gaussian SDR performs similarly to csMAVE when $p=5$, but is outperformed by csMAVE when $p=10$ and $20$. Compared with Model I, the advantage of FlowSDR over neural Gaussian SDR becomes more pronounced in Model II, reflecting the robustness of FlowSDR to violations of the Gaussian distribution assumption.

In Model III, we consider a two-component mixture response model. For $X=(X_1,\ldots,X_{20})$, let $X\sim \mathcal{N}(0,I_{20})$. Let
\[
  Y \mid X \sim 0.3\,\mathcal{N}(\beta_3^\top X,1)+0.7\,\mathcal{N}(\beta_4^\top X,1),
\]
where $\beta_3=(1,0.5,0,\ldots,0)^\top$ and $\beta_4=(-1,0,0.5,0,\ldots,0)^\top$. The central subspace is $\operatorname{span}(\beta_3,\beta_4)$.

\begin{table}[t]
  \centering
  \scriptsize
  \setlength{\tabcolsep}{3.5pt}
  \caption{Results for Model III with $p=20$ and varying sample size $n$. The mean, median, and standard deviation of $\Delta_{\rm op}$ and $\Delta_{\rm F}$ over $100$ replications are reported.}
  \label{tab:mixture-comparison}
  \begin{tabular}{clcccc}
    \toprule
    $n$ & Method & $\Delta_{\rm op}$ mean $\pm$ SD & $\Delta_{\rm op}$ median & $\Delta_{\rm F}$ mean $\pm$ SD & $\Delta_{\rm F}$  median \\
    \midrule
    2000 & NN-SDR & $0.9791 \pm 0.0401$ & $0.9969$ & $1.4412 \pm 0.0576$ & $1.4583$ \\
     & BENN-SDR & $0.9955 \pm 0.0109$ & $0.9983$ & $1.4305 \pm 0.0173$ & $1.4321$ \\
     & MAVE & $0.9391 \pm 0.0759$ & $0.9724$ & $1.3696 \pm 0.1039$ & $1.4131$ \\
     & csMAVE & $0.8161 \pm 0.1700$ & $0.8774$ & $1.1821 \pm 0.2373$ & $1.2829$ \\
     & \textbf{neural Gaussian SDR} & $\mathbf{0.4804 \pm 0.1284}$ & $\mathbf{0.4548}$ & $\mathbf{0.7063 \pm 0.1772}$ & $\mathbf{0.6674}$ \\
     & FlowSDR & $0.6334 \pm 0.1784$ & $0.6155$ & $0.9207 \pm 0.2468$ & $0.8869$ \\
     \hline
    \addlinespace[2pt]
    5000 & NN-SDR & $0.9878 \pm 0.0295$ & $0.9959$ & $1.4173 \pm 0.0408$ & $1.4262$ \\
     & BENN-SDR & $0.9982 \pm 0.0019$ & $0.9988$ & $1.4193 \pm 0.0040$ & $1.4196$ \\
     & MAVE & $0.9542 \pm 0.0551$ & $0.9745$ & $1.3654 \pm 0.0775$ & $1.3955$ \\
     & csMAVE & $0.5549 \pm 0.2326$ & $0.4819$ & $0.7987 \pm 0.3258$ & $0.7005$ \\
     & neural Gaussian SDR & $0.2963 \pm 0.0694$ & $0.2851$ & $0.4340 \pm 0.0954$ & $0.4220$ \\
     & \textbf{FlowSDR} & $\mathbf{0.2875 \pm 0.0618}$ & $\mathbf{0.2824}$ & $\mathbf{0.4230 \pm 0.0842}$ & $\mathbf{0.4154}$ \\
        \hline
    \addlinespace[2pt]
    10000 & NN-SDR & $0.9921 \pm 0.0150$ & $0.9945$ & $1.4122 \pm 0.0206$ & $1.4144$ \\
     & BENN-SDR & $0.9983 \pm 0.0012$ & $0.9985$ & $1.4154 \pm 0.0021$ & $1.4155$ \\
     & MAVE & $0.9347 \pm 0.0829$ & $0.9665$ & $1.3302 \pm 0.1169$ & $1.3751$ \\
     & csMAVE & $0.3411 \pm 0.1851$ & $0.2860$ & $0.4937 \pm 0.2593$ & $0.4189$ \\
     & neural Gaussian SDR & $0.2109 \pm 0.0504$ & $0.2021$ & $0.3092 \pm 0.0691$ & $0.2979$ \\
     & \textbf{FlowSDR} & $\mathbf{0.1905 \pm 0.0344}$ & $\mathbf{0.1902}$ & $\mathbf{0.2809 \pm 0.0468}$ & $\mathbf{0.2778}$ \\
    \bottomrule
  \end{tabular}
\end{table}

From Table ~\ref{tab:mixture-comparison}, FlowSDR and neural Gaussian SDR substantially outperform the competing SDR methods across all sample sizes. When $n=2000$, neural Gaussian SDR achieves the best performance, while FlowSDR becomes the top-performing method with $n=5000$ and $10000$. While csMAVE, FlowSDR, and neural Gaussian SDR improve as $n$ increases, MAVE, NN-SDR, and BENN-SDR remain close to random guessing in this multimodal setting.

In Model IV, we consider a conditional distribution where the tail behavior is not captured by the mean–variance structure. For $X=(X_1,\ldots,X_{20})$, let $X\sim \mathcal{N}(0,I_{20})$. 
Let $W=\beta_5^\top X$, where
$\beta_5=(1,1,0,\ldots,0)^\top/\sqrt{2}$. We generate the response as
\[
  Y\mid X \sim
  \begin{cases}
    \mathcal{N}(0,1), & \mbox{ if }W>0,\\
    t_3/\sqrt{3}, & \mbox{ if }W\leq 0.
  \end{cases}
\]
Both regimes have mean zero and variance one, so the sufficient direction enters only through the tail shape of the conditional distribution. The central subspace is $\operatorname{span}(\beta_5)$.

\begin{table}[t]
  \centering
  \scriptsize
  \setlength{\tabcolsep}{3.5pt}
  \caption{Results for Model IV with $p=20$ and varying sample size $n$. The mean, median, and standard deviation of $\Delta_{\rm op}$ and $\Delta_{\rm F}$ over $100$ replications are reported.}
  \label{tab:sign-tail-comparison}
  \begin{tabular}{clcccc}
    \toprule
    $n$ & Method & $\Delta_{\rm op}$ mean $\pm$ SD & $\Delta_{\rm op}$ median & $\Delta_{\rm F}$ mean $\pm$ SD & $\Delta_{\rm F}$ median \\
    \midrule
    2000 & NN-SDR & $0.9315 \pm 0.0375$ & $0.9315$ & $1.3173 \pm 0.0531$ & $1.3173$ \\
     & BENN-SDR & $0.9224 \pm 0.0478$ & $0.9257$ & $1.3044 \pm 0.0676$ & $1.3092$ \\
     & MAVE & $0.9749 \pm 0.0382$ & $0.9912$ & $1.3788 \pm 0.0541$ & $1.4018$ \\
     & csMAVE & $0.8939 \pm 0.1268$ & $0.9446$ & $1.2642 \pm 0.1793$ & $1.3359$ \\
     & neural Gaussian SDR & $0.9678 \pm 0.0522$ & $0.9858$ & $1.3687 \pm 0.0738$ & $1.3941$ \\
     & \textbf{FlowSDR} & $\mathbf{0.5733 \pm 0.1692}$ & $\mathbf{0.5635}$ & $\mathbf{0.8108 \pm 0.2393}$ & $\mathbf{0.7969}$ \\
        \hline
    \addlinespace[2pt]
    5000 & NN-SDR & $0.9138 \pm 0.0496$ & $0.9185$ & $1.2923 \pm 0.0702$ & $1.2990$ \\
     & BENN-SDR & $0.9124 \pm 0.0509$ & $0.9190$ & $1.2904 \pm 0.0720$ & $1.2997$ \\
     & MAVE & $0.9697 \pm 0.0466$ & $0.9843$ & $1.3714 \pm 0.0659$ & $1.3921$ \\
     & csMAVE & $0.7584 \pm 0.2246$ & $0.8367$ & $1.0726 \pm 0.3177$ & $1.1833$ \\
     & neural Gaussian SDR & $0.9513 \pm 0.0572$ & $0.9715$ & $1.3454 \pm 0.0810$ & $1.3739$ \\
     & \textbf{FlowSDR} & $\mathbf{0.2512 \pm 0.0789}$ & $\mathbf{0.2371}$ & $\mathbf{0.3553 \pm 0.1116}$ & $\mathbf{0.3353}$ \\
        \hline
    \addlinespace[2pt]
    10000 & NN-SDR & $0.9180 \pm 0.0414$ & $0.9235$ & $1.2982 \pm 0.0586$ & $1.3060$ \\
     & BENN-SDR & $0.9233 \pm 0.0413$ & $0.9225$ & $1.3057 \pm 0.0584$ & $1.3046$ \\
     & MAVE & $0.9731 \pm 0.0396$ & $0.9905$ & $1.3762 \pm 0.0560$ & $1.4008$ \\
     & csMAVE & $0.4035 \pm 0.2684$ & $0.2853$ & $0.5706 \pm 0.3796$ & $0.4034$ \\
     & neural Gaussian SDR & $0.9587 \pm 0.0574$ & $0.9841$ & $1.3558 \pm 0.0812$ & $1.3917$ \\
     & \textbf{FlowSDR} & $\mathbf{0.1449 \pm 0.1080}$ & $\mathbf{0.1245}$ & $\mathbf{0.2049 \pm 0.1528}$ & $\mathbf{0.1761}$ \\
    \bottomrule
  \end{tabular}
\end{table}

From Table ~\ref{tab:sign-tail-comparison}, we clearly observe the distinction between FlowSDR and neural Gaussian SDR.    Neural Gaussian SDR remains close to random guessing 
as  $n$ increases. This is expected because neither the conditional mean nor the conditional variance of $Y\mid X$ captures the SDR direction. Although csMAVE improves substantially when $n=10000$,  FlowSDR is the only method that reliably recovers the tail direction, reducing the mean operator error from $0.5733$ at $n=2000$ to $0.1449$ at $n=10000$.

\subsection{UTKFace age analysis}
We evaluate FlowSDR on the UTKFace face-age dataset, where each image filename encodes age, gender, ethnicity, and timestamp. The response variable is the continuous age label. Our analysis is feature-based rather than end-to-end image training. From the aligned and cropped images, we extract HOG descriptors  \citep{dalal2005hog}  and RGB color histograms, yielding $1812$ features for each of $23705$ usable images. We use a single stratified split consisting of 70\% training, 15\% validation, and 15\% test samples, where the stratification is based on age bins.
The validation set is used for early stopping and model selection during neural network training, while the test set is reserved exclusively for 
final evaluation. We provide the implementation details of this analysis in Appendix D.

Across different methods, we report test data MAE, RMSE and NLL in Table~\ref{tab:utkface-hog-color-results}. In addition to the methods compared in the simulation studies, we also include PLS \citep{wold2001pls} and the full MLP Gaussian method, which corresponds to neural Gaussian SDR without dimension reduction. 
MAVE and csMAVE are not designed to handle large predictor dimensions directly. We therefore apply marginal correlation screening to select the top $100$ features from $X \in \mathbb{R}^{1812}$, and then fit MAVE and csMAVE using the screened features. 

Since the true structural dimension $d$ of the central subspace is unknown, Table~\ref{tab:utkface-hog-color-results} include $d \in \{2,4,6,8\}$ as the working structural dimension. 
FlowSDR achieves the strongest overall prediction performance. Specifically, the best MAE  $7.5672$ is achieved by FlowSDR with $d=8$, and the best NLL  $0.4547$ is achieved by FlowSDR with $d=4$. While the best RMSE $10.8607$ is achieved by BENN-SDR with $d=8$, the RMSE values based on FlowSDR are not much larger. 

One advantage of the proposed likelihood-based framework is that it directly estimates the conditional density $f(Y | B^\top X)$. This naturally facilitates the construction of conformal prediction intervals \citep{vovk2005algorithmic,lei2018distribution}. A key component of conformal prediction is the use of a calibration set to estimate the conformal quantile required for finite-sample distribution-free coverage on future test observations. The detailed development of conformal prediction for FlowSDR and neural Gaussian SDR is provided in Appendix~\ref{sec:ci}.

\begin{table}[t]
  \centering
  \scriptsize
  \setlength{\tabcolsep}{3.5pt}
    \caption{UTKFace age prediction results.}
  \label{tab:utkface-hog-color-results}
  \begin{tabular}{llccc}
    \toprule
    Method & $d$ & MAE  & RMSE  & NLL \\
    \midrule
    Mean age & -- & 15.3632 & 19.9325 & 1.4220 \\
    \addlinespace[2pt]
    Full MLP Gaussian & 1812 & 9.1305 & 12.6048 & 0.8202 \\
    \addlinespace[2pt]
    PLS & 2 & 11.0413 & 14.2752 & 1.0887 \\
    PLS & 4 & 10.1378 & 13.2021 & 1.0115 \\
    PLS & 6 & 9.8602 & 12.8915 & 0.9885 \\
    PLS & 8 & 9.7677 & 12.7086 & 0.9744 \\
    \addlinespace[2pt]
    MAVE (screened) & 2 & 10.2714 & 14.0034 & 0.9599 \\
    MAVE (screened) & 4 & 10.3170 & 14.1961 & 0.9651 \\
    MAVE (screened) & 6 & 10.1772 & 14.1088 & 0.9669 \\
    MAVE (screened) & 8 & 10.0990 & 13.9373 & 0.9467 \\
    \addlinespace[2pt]
    csMAVE (screened) & 2 & 9.9563 & 13.9294 & 0.9221 \\
    csMAVE (screened) & 4 & 10.1548 & 13.9639 & 0.9228 \\
    csMAVE (screened) & 6 & 9.9163 & 13.7077 & 0.9549 \\
    csMAVE (screened) & 8 & 9.8385 & 13.7456 & 0.9415 \\
    \addlinespace[2pt]
    NN-SDR & 2 & 8.1268 & 11.4577 & 0.7246 \\
    NN-SDR & 4 & 7.8089 & 11.0967 & 0.6673 \\
    NN-SDR & 6 & 7.6483 & 10.8639 & 0.6636 \\
    NN-SDR & 8 & 7.7730 & 11.1074 & 0.7098 \\
    \addlinespace[2pt]
    BENN-SDR & 2 & 8.1844 & 11.3808 & 0.7068 \\
    BENN-SDR & 4 & 7.8504 & 10.9783 & 0.6397 \\
    BENN-SDR & 6 & 7.7110 & 10.9905 & 0.7004 \\
    BENN-SDR & 8 & 7.6103 & \textbf{10.8607} & 0.7362 \\
    \addlinespace[2pt]
    neural Gaussian SDR & 2 & 8.2245 & 11.6220 & 0.7047 \\
    neural Gaussian SDR & 4 & 7.9063 & 11.3618 & 0.7357 \\
    neural Gaussian SDR & 6 & 8.1976 & 11.6335 & 0.7812 \\
    neural Gaussian SDR & 8 & 8.0622 & 11.5171 & 0.8557 \\
    \addlinespace[2pt]
    FlowSDR & 2 & 7.7877 & 11.4942 & 0.4742 \\
    FlowSDR & 4 & 7.5981 & 11.2485 & \textbf{0.4547} \\
    FlowSDR & 6 & 7.7559 & 11.4290 & 0.5287 \\
    FlowSDR & 8 & \textbf{7.5672} & 11.2708 & 0.5224 \\
    \bottomrule
  \end{tabular}
\end{table}

We use a single stratified split consisting of 70\% training, 10\% validation, 10\% calibration, and 10\% test samples, and compare prediction interval performance on the test set at the 95\% nominal level with working structural dimension $d=4$. Prediction intervals based on neural Gaussian SDR achieve an empirical coverage of 93.8\% with average interval length $42.53$. FlowSDR achieves improved empirical coverage of 94.5\% together with a shorter average interval length of $37.19$.

\section{Concluding remarks}

We develop a novel likelihood-based framework for joint representation learning and conditional density estimation. Superior SDR performance is demonstrated across a range of synthetic examples, while the estimated conditional density naturally enables the construction of conformal prediction intervals.

Conditional normalizing flows play a central role in this framework by transporting the unknown conditional distribution to a fixed base distribution. Although we focus on linear SDR with a univariate response in this paper, the proposed framework can be naturally extended to nonlinear SDR with multivariate responses. Specifically, for predictor \(X \in \mathbb{R}^p\) and response \(Y \in \mathbb{R}^q\), the goal is to learn a low-dimensional nonlinear representation \(\tau(X)\) satisfying
\[
Y \indep X \mid \tau(X).
\]
The corresponding nonlinear FlowSDR objective becomes
\[
\max_{\tau,\theta}
{\rm E}\left\{
\log f_\theta(Y \mid \tau(X))
\right\},
\]
where \(f_\theta(Y \mid \tau(X))\) is modeled using a multivariate conditional normalizing flow.

Since normalizing flows are closely connected to stochastic interpolation and transport-based generative modeling \citep{albergo2023building,albergo2025stochastic,xu2025conditional}, sufficient dimension reduction along this direction warrants further investigation.

\bibliographystyle{plainnat}
\bibliography{references}

@article{wold2001pls,
  title={{PLS}-regression: a basic tool of chemometrics},
  author={Wold, Svante and Sj{\"o}str{\"o}m, Michael and Eriksson, Lennart},
  journal={Chemometrics and Intelligent Laboratory Systems},
  volume={58},
  number={2},
  pages={109--130},
  year={2001}
}

@book{vovk2005algorithmic,
  title={Algorithmic Learning in a Random World},
  author={Vovk, Vladimir and Gammerman, Alexander and Shafer, Glenn},
  year={2005},
  publisher={Springer}
}

@article{lei2018distribution,
  title={Distribution-Free Predictive Inference for Regression},
  author={Lei, Jing and G'Sell, Max and Rinaldo, Alessandro and Tibshirani, Ryan and Wasserman, Larry},
  journal={Journal of the American Statistical Association},
  volume={113},
  number={523},
  pages={1094--1111},
  year={2018}
}

@inproceedings{dalal2005hog,
  title={Histograms of oriented gradients for human detection},
  author={Dalal, Navneet and Triggs, Bill},
  booktitle={IEEE Computer Society Conference on Computer Vision and Pattern Recognition},
  year={2005}
}

@inproceedings{irons2022triangular,
  title     = {Triangular Flows for Generative Modeling: Statistical Consistency, Smoothness Classes, and Fast Rates},
  author    = {Irons, Nicholas J. and Scetbon, Meyer and Pal, Soumik and Harchaoui, Zaid},
  booktitle = {International Conference on Artificial Intelligence and Statistics (AISTATS)},
  year      = {2022}
}

@article{cook2009likelihood,
  author  = {Cook, R. Dennis and Forzani, Liliana},
  title   = {Likelihood-based sufficient dimension reduction},
  journal = {Journal of the American Statistical Association},
  volume  = {104},
  number  = {485},
  pages   = {197--208},
  year    = {2009}
}

@article{edelman1998geometry,
  author  = {Edelman, Alan and Arias, Tom{\'a}s A. and Smith, Steven T.},
  title   = {The Geometry of Algorithms with Orthogonality Constraints},
  journal = {SIAM Journal on Matrix Analysis and Applications},
  year    = {1998},
  volume  = {20},
  number  = {2},
  pages   = {303--353}
}

@article{burges2010dimension,
  title={Dimension reduction: A guided tour},
  author={Burges, Christopher J. C.},
  journal={Foundations and Trends{\textregistered} in Finance},
  volume={2},
  number={4},
  pages={275--365},
  year={2010}
}

@article{dinh2014nice,
  title={NICE: Non-linear Independent Components Estimation},
  author={Dinh, Laurent and Krueger, David and Bengio, Yoshua},
  journal={arXiv preprint arXiv:1410.8516},
  year={2014}
}

@inproceedings{stirn2023faithful,
  title={Faithful Heteroscedastic Regression with Neural Networks},
  author={Stirn, Andrew and Wessels, Harm and Schertzer, Megan and Pereira, Laura and Sanjana, Neville and Knowles, David},
  booktitle={International Conference on Artificial Intelligence and Statistics},
  year={2023}
}

@inproceedings{belghazi2018mine,
  title     = {Mutual Information Neural Estimation},
  author    = {Belghazi, Mohamed Ishmael and Baratin, Aristide and Rajeshwar, Sai and Ozair, Sherjil and Bengio, Yoshua and Courville, Aaron and Hjelm, Devon},
  booktitle = {International Conference on Machine Learning},
  year      = {2018}
}

@article{nguyen2010estimating,
  title   = {Estimating divergence functionals and the likelihood ratio by convex risk minimization},
  author  = {Nguyen, XuanLong and Wainwright, Martin J. and Jordan, Michael I.},
  journal = {IEEE Transactions on Information Theory},
  volume  = {56},
  number  = {11},
  pages   = {5847--5861},
  year    = {2010},
  publisher = {IEEE}
}

@article{oord2018cpc,
  title   = {Representation Learning with Contrastive Predictive Coding},
  author  = {van den Oord, Aaron and Li, Yazhe and Vinyals, Oriol},
  journal = {arXiv preprint arXiv:1807.03748},
  year    = {2018}
}

@article{li1991sir,
  author = {Li, Ker-Chau},
  title = {Sliced Inverse Regression for Dimension Reduction},
  journal = {Journal of the American Statistical Association},
  year = {1991},
  volume = {86},
  number = {414},
  pages = {316--327}
}

@article{cook1991save,
  author = {Cook, R. Dennis and Weisberg, Sanford},
  title = {Discussion of {S}liced {I}nverse {R}egression for {D}imension {R}eduction},
  journal = {Journal of the American Statistical Association},
  year = {1991},
  volume = {86},
  number = {414},
  pages = {328--332}
}

@article{xia2002mave,
  author = {Xia, Yingcun and Tong, Howell and Li, W. K. and Zhu, Li-Xing},
  title = {An Adaptive Estimation of Dimension Reduction Space},
  journal = {Journal of the Royal Statistical Society Series B: Statistical Methodology},
  year = {2002},
  volume = {64},
  number = {3},
  pages = {363--410}
}

@inproceedings{dinh2017realnvp,
  author = {Dinh, Laurent and Sohl-Dickstein, Jascha and Bengio, Samy},
  title = {Density Estimation Using Real {NVP}},
  booktitle = {International Conference on Learning Representations},
  year = {2017}
}

@inproceedings{durkan2019neural,
  author = {Durkan, Conor and Bekasov, Artur and Murray, Iain and Papamakarios, George},
  title = {Neural Spline Flows},
  booktitle = {Advances in Neural Information Processing Systems},
  year = {2019}
}

@inproceedings{papamakarios2017maf,
  title     = {Masked Autoregressive Flow for Density Estimation},
  author    = {Papamakarios, George and Pavlakou, Theo and Murray, Iain},
  booktitle = {Advances in Neural Information Processing Systems},
  volume    = {30},
  year      = {2017}
}

@article{papamakarios2021normalizing,
  author = {Papamakarios, George and Nalisnick, Eric and Rezende, Danilo Jimenez and Mohamed, Shakir and Lakshminarayanan, Balaji},
  title = {Normalizing Flows for Probabilistic Modeling and Inference},
  journal = {Journal of Machine Learning Research},
  year = {2021},
  volume = {22},
  number = {57},
  pages = {1--64}
}

@article{zhang2020maximum,
  title   = {The maximum separation subspace in sufficient dimension reduction with categorical response},
  author  = {Zhang, Xin and Mai, Qing and Zou, Hui},
  journal = {Journal of Machine Learning Research},
  volume  = {21},
  number  = {29},
  pages   = {1--36},
  year    = {2020}
}

@article{chen2024deep,
  title   = {Deep nonlinear sufficient dimension reduction},
  author  = {Chen, Yinfeng and Jiao, Yuling and Qiu, Rui and Yu, Zhou},
  journal = {The Annals of Statistics},
  volume  = {52},
  number  = {3},
  pages   = {1201--1226},
  year    = {2024}
}

@article{yin2008successive,
  title   = {Successive direction extraction for estimating the central subspace in a multiple-index regression},
  author  = {Yin, Xiangrong and Li, Bing and Cook, R. Dennis},
  journal = {Journal of Multivariate Analysis},
  volume  = {99},
  number  = {8},
  pages   = {1733--1757},
  year    = {2008}
}

@inproceedings{liang2022nonlinear,
  author = {Liang, Siqi and Sun, Yan and Liang, Faming},
  title = {Nonlinear sufficient dimension reduction with a stochastic neural network},
  booktitle = {Advances in Neural Information Processing Systems},
  year = {2022}
}

@inproceedings{zhang2019learning,
  author = {Zhang, Guannan and Zhang, Jiaxin and Hinkle, Jacob},
  title = {Learning nonlinear level sets for dimensionality reduction in function approximation},
  booktitle = {Advances in Neural Information Processing Systems},
  year = {2019}
}

@inproceedings{meng2020sufficient,
  author = {Meng, Cheng and Yu, Jun and Zhang, Jingyi and Ma, Ping and Zhong, Wenxuan},
  title = {Sufficient dimension reduction for classification using principal optimal transport direction},
  booktitle = {Advances in Neural Information Processing Systems},
  year = {2020}
}

@article{yang2025golden,
  title   = {Golden Ratio-Based Sufficient Dimension Reduction},
  author  = {Yang, Wenjing and Yang, Yuhong},
  journal = {IEEE Transactions on Information Theory},
    volume  = {71},
  number  = {8},
  year    = {2025}
}

@article{yin2002dimension,
  title={Dimension reduction for the conditional k th moment in regression},
  author={Yin, Xiangrong and Cook, R. Dennis},
  journal={Journal of the Royal Statistical Society Series B: Statistical Methodology},
  volume={64},
  number={2},
  pages={159--175},
  year={2002}
}

@article{kapla2022fusing,
  title   = {Fusing sufficient dimension reduction with neural networks},
  author  = {Kapla, Daniel and Fertl, Lukas and Bura, Efstathia},
  journal = {Computational Statistics \& Data Analysis},
  volume  = {168},
  pages   = {107390},
  year    = {2022}
}

@article{tang2026belted,
  title   = {Belted and ensembled neural network for linear and nonlinear sufficient dimension reduction},
  author  = {Tang, Yin and Li, Bing},
  journal = {Journal of the American Statistical Association},
   note    = {Forthcoming},
  year    = {2026},
  doi     = {10.1080/01621459.2025.2590775}
}

@article{fukumizu2004dimensionality,
  title   = {Dimensionality reduction for supervised learning with Reproducing Kernel Hilbert Spaces},
  author  = {Fukumizu, Kenji and Bach, Francis R. and Jordan, Michael I.},
  journal = {Journal of Machine Learning Research},
  volume  = {5},
  pages={73--99},  
  year    = {2004}
}

@article{fukumizu2009kernel,
  title   = {Kernel dimension reduction in regression},
  author  = {Fukumizu, Kenji and Bach, Francis R. and Jordan, Michael I.},
  journal = {The Annals of Statistics},
  volume  = {37},
  number  = {4},
  pages   = {1871--1905},
  year    = {2009}
}

@inproceedings{wu2008localized,
  author = {Wu, Qiang and Mukherjee, Sayan and Liang, Feng},
  title = {Localized sliced inverse regression},
  booktitle = {Advances in Neural Information Processing Systems},
  year = {2008}
}

@article{yeh2008nonlinear,
  title   = {Nonlinear dimension reduction with kernel sliced inverse regression},
  author  = {Yeh, Yi-Ren and Huang, Su-Yun and Lee, Yuh-Jye},
  journal = {IEEE Transactions on Knowledge and Data Engineering},
  volume  = {21},
  number  = {11},
  pages   = {1590--1603},
  year    = {2008}
}

@article{cai2020online,
  title   = {Online sufficient dimension reduction through sliced inverse regression},
  author  = {Cai, Zhanrui and Li, Runze and Zhu, Liping},
  journal = {Journal of Machine Learning Research},
  volume  = {21},
  number  = {10},
  pages   = {1--25},
  year    = {2020}
}

@article{artemiou2021realtime,
  title   = {Real-time sufficient dimension reduction through principal least squares support vector machines},
  author  = {Artemiou, Andreas and Dong, Yuexiao and Shin, Seung Jun},
  journal = {Pattern Recognition},
  volume  = {112},
  pages   = {107768},
  year    = {2021}
}

@article{wang2008sliced,
  title   = {Sliced regression for dimension reduction},
  author  = {Wang, Hansheng and Xia, Yingcun},
  journal = {Journal of the American Statistical Association},
  volume  = {103},
  number  = {482},
  pages   = {811--821},
  year    = {2008}
}

@article{hsing2009rkhs,
  title   = {An {RKHS} formulation of the inverse regression dimension-reduction problem},
  author  = {Hsing, Tailen and Ren, Haobo},
  journal = {The Annals of Statistics},
    volume  = {37},
  number  = {2},
  pages   = {726--755},
  year    = {2009}
}

@article{li2017nonlinear,
  title   = {Nonlinear sufficient dimension reduction for functional data},
  author  = {Li, Bing and Song, Jun},
    journal = {The Annals of Statistics},
     volume  = {45},
  number  = {3},
  pages   = {1059--1095},
  year    = {2017}
}

@article{tang2020highdim,
  title   = {High-dimensional interactions detection with sparse principal hessian matrix},
  author  = {Tang, Cheng Yong and Fang, Ethan X. and Dong, Yuexiao},
  journal = {Journal of Machine Learning Research},
  volume  = {21},
  number  = {19},
  pages   = {1--25},
  year    = {2020}
}

@article{luo2022efficient,
  title   = {On efficient dimension reduction with respect to the interaction between two response variables},
  author  = {Luo, Wei},
  journal = {Journal of the Royal Statistical Society Series B: Statistical Methodology},
  volume  = {84},
  number  = {2},
  pages   = {269--294},
  year    = {2022}
}

@article{tangkaratt2014conditional,
  title   = {Conditional density estimation with dimensionality reduction via squared-loss conditional entropy minimization},
  author  = {Tangkaratt, Voot and Xie, Ning and Sugiyama, Masashi},
  journal = {Neural Computation},
  volume  = {27},
  number  = {1},
  pages   = {228--254},
  year    = {2014}
}

@article{kobyzev2020normalizing,
  title   = {Normalizing flows: An introduction and review of current methods},
  author  = {Kobyzev, Ivan and Prince, Simon J. D. and Brubaker, Marcus A.},
  journal = {IEEE Transactions on Pattern Analysis and Machine Intelligence},
  volume  = {43},
  number  = {11},
  pages   = {3964--3979},
  year    = {2020}
}

@article{albergo2025stochastic,
  title   = {Stochastic interpolants: A unifying framework for flows and diffusions},
  author  = {Albergo, Michael and Boffi, Nicholas M. and Vanden-Eijnden, Eric},
  journal = {Journal of Machine Learning Research},
  volume  = {26},
  number  = {209},
  pages   = {1--80},
  year    = {2025}
}

@inproceedings{albergo2023building,
  author    = {Albergo, Michael and Vanden-Eijnden, Eric},
  title     = {Building Normalizing Flows with Stochastic Interpolants},
  booktitle = {International Conference on Learning Representations},
  year      = {2023}
}

@article{xu2025conditional,
  title={On Conditional Stochastic Interpolation for Generative Nonlinear Sufficient Dimension Reduction},
  author={Xu, Shuntuo and Yu, Zhou and Huang, Jian},
  journal={arXiv preprint arXiv:2512.18971},
  year={2025}
}

@book{li2018sdr,
  title     = {Sufficient Dimension Reduction: Methods and Applications with R},
  author    = {Li, Bing},
  publisher = {Chapman and Hall/CRC},
  year      = {2018}
}

@article{cook2007dimension,
  title   = {Fisher Lecture: Dimension Reduction in Regression},
  author  = {Cook, R. Dennis},
  journal = {Statistical Science},
  volume  = {22},
  number  = {1},
  pages   = {1--26},
  year    = {2007}
}

@article{cook2026foundations,
  title   = {On the Foundational Arguments of Sufficient Dimension Reduction},
  author  = {Cook, R. Dennis},
  journal = {Wiley Interdisciplinary Reviews: Computational Statistics},
  year    = {2026},
    note    = {Forthcoming},
      doi     = {10.1002/wics.70064}
}

@article{luo2020matching,
  title   = {Matching Using Sufficient Dimension Reduction for Causal Inference},
  author  = {Luo, Wei and Zhu, Yeying},
  journal = {Journal of Business \& Economic Statistics},
  volume  = {38},
  number  = {4},
  pages   = {888--900},
  year    = {2020}
}

@article{li2011psvm,
  title   = {Principal Support Vector Machines for Linear and Nonlinear Sufficient Dimension Reduction},
  author  = {Li, Bing and Artemiou, Andreas and Li, Lexin},
  journal = {Annals of Statistics},
  volume  = {39},
  number  = {6},
  pages   = {3182--3210},
  year    = {2011}
}

@article{ding2020double,
  title   = {Double-slicing assisted sufficient dimension reduction for high-dimensional censored data},
  author  = {Ding, Shanshan and Qian, Wei and Wang, Lan},
  journal = {The Annals of Statistics},
  volume  = {48},
  number  = {4},
  pages   = {2132--2154},
  year    = {2020}
}

\newpage
\appendix
\section*{Appendix}

\section{Fisher consistency of the information index maximizer}
\label{sec:fisher info}

First we formally define the central subspace of the response $Y$ versus the predictor $X$. Recall from (\ref{eq:sdr}) that 
  $$Y \indep X \mid B^{\top}X.$$
  The conditional independence above holds when $B$ is replaced with any matrix whose columns form a basis of $\operatorname{span}(B)$. We refer to $\operatorname{span}(B)$ as the dimension reduction subspace (DRS) for the regression of $Y$ on $X$. Let $\mathcal{S}_{Y|X}$ denote the intersection of all DRS's. Under the conditions of Proposition 5 in  \cite{yin2008successive}, $\mathcal{S}_{Y|X}$ remains a DRS, which we refer to as the central subspace (CS).  We assume that the CS exists throughout this Appendix. The dimension $d$ of  $\mathcal{S}_{Y|X}$ is called the structural dimension of the regression. 

The following Lemma is parallel to Proposition 1 of \cite{yin2008successive}.

\begin{lemma}[Ordering and characterization properties of the information index]
\label{lem:order}
Let $A_j$ be a $p \times q_j$ matrix, $j = 1,2$, and let $A$ be a $p \times q$ matrix. For the  information index $\mathcal{J}(A)$ 
defined in (\ref{eq:mi}):
$$ \mathcal{J}(A)={\rm E}\left\{ \log{f(Y | A^\top X)} \right\},$$
 the following properties hold:

\begin{enumerate}
\item If $\operatorname{span}(A) \subseteq \operatorname{span}(A_1)$, then $\mathcal{J}(A) \le \mathcal{J}(A_1)$.  
If $\operatorname{span}(A) = \operatorname{span}(A_1)$, then $\mathcal{J}(A) = \mathcal{J}(A_1)$.

\item $\mathcal{J}(A) = \mathcal{J}(I_p)$ if and only if $Y \indep X \mid A^\top X$.

\item $\mathcal{J}(A) \ge 0$ for all $A$, and $\mathcal{J}(A) = 0$ if and only if $Y \indep A^\top X$.

\item Let $\alpha$ be a basis matrix for the central subspace $\mathcal{S}_{Y \mid X}$. Then
\[
\mathcal{J}(\alpha) \ge \mathcal{J}(A),
\]
with equality if and only if $\operatorname{span}(A) = \operatorname{span}(\alpha)$.

\item If $q_1 < q_2 < d$, where $d = \dim(\mathcal{S}_{Y \mid X})$, then
\[
\mathcal{J}(I_p) = \mathcal{J}(\alpha) > \max \mathcal{J}(A_{q_2}) > \max \mathcal{J}(A_{q_1}),
\]
where the maxima are taken over all $p \times q_j$ matrices $A_{q_j}$.
\end{enumerate}
\end{lemma}

\noindent {\bf Proof}: \cite{yin2008successive} defined information index 
$$\mathcal{I}(A)={\rm E}\left\{ \log\frac{f(Y | A^\top X)}{f(Y)} \right\},$$
which coincides with the mutual information between $Y$ and $A^\top X$. It follows that 
$$\mathcal{J}(A)=\mathcal{I}(A)+{\rm E}\left\{ \log{f(Y)} \right\}.$$
Since ${\rm E}\{\log f(Y)\}$ does not depend on $A$, the ordering and maximization properties of $\mathcal{J}(A)$ follow directly from Proposition 1 of \cite{yin2008successive}, where the corresponding results are established for $\mathcal{I}(A)$. 
 \qed

Part 4 of Lemma \ref{lem:order} implies that searching over $p \times d$ matrices yields a basis for the central subspace (CS). Part 5 further suggests that the CS can be recovered at the population level via multi-dimensional optimization. In particular, by considering successively larger subspaces, a basis for the CS is obtained when the maximum is attained for the first time.

Let $\mathcal V_{p,d}
=
\{ B \in \mathbb{R}^{p\times d} : B^T B=I_d\}$ denote the Stiefel manifold. The Fisher consistency of the information index maximizer follows directly from Lemma \ref{lem:order}. Therefore we state the result  without proof.  

\begin{theorem}[Fisher consistency]
\label{thm:fisher}
For $B_0=\arg\max_{B \in \mathcal V_{p,d}} \mathcal{J}(B)$ defined in (\ref{eq:population}),  
we have $\operatorname{span}(B_0)=\mathcal{S}_{Y|X}$.
\end{theorem}

\section{Consistency of the FlowSDR estimator}
\label{sec:flowsdr consistency}

Throughout the main text, $\Theta$ denotes a generic conditional flow class. For asymptotic analysis, we consider a sequence of sieve spaces $\{\Theta_n\}$ that increases in richness with sample size $n$. 

\begin{lemma}[Uniform convergence over sieve class]
\label{lem:ulln_sieve}

Let $\Theta_n$ denote the conditional flow class at sample size $n$. 
For each $B \in \mathcal V_{p,d}$, define
\begin{equation}
\label{eq:sieve1}
M_n(B)
=
\sup_{\theta \in \Theta_n}
\frac{1}{n}
\sum_{i=1}^n
\log f_\theta(Y_i | B^\top X_i),
\end{equation}
and its population counterpart
\begin{equation}
\label{eq:sieve2}
M_n^{\mathrm{pop}}(B)
=
\sup_{\theta \in \Theta_n}
{\rm E}\bigl\{\log f_\theta(Y | B^\top X)\bigr\}.
\end{equation}

Assume that the class
\[
\bigl\{
(x,y) \mapsto \log f_\theta(y | B^\top x)
:
\theta \in \Theta_n,\;
B \in \mathcal V_{p,d}
\bigr\}
\]
satisfies a uniform law of large numbers (ULLN). Then
\[
\sup_{B \in \mathcal V_{p,d}}
\left|
M_n(B) - M_n^{\mathrm{pop}}(B)
\right|
\overset{p}{\longrightarrow} 0.
\]

\end{lemma}

\noindent {\bf Proof}:
Let $\mathbb P_n$ denote the empirical measure such that
$\mathbb P_n g = \frac{1}{n}\sum_{i=1}^n g(X_i,Y_i)$. 
Let $\mathbb P g = {\rm E}\{g(X,Y)\}$ denote the corresponding population expectation. Note that $\log f_\theta(Y | B^\top X)$ is a measurable function of $(X,Y)$, so the uniform law of large numbers applies directly to this function class despite the conditional form of the density.

For any fixed $B\in\mathcal V_{p,d}$, let
\[
a_{n,\theta}(B)=
\mathbb P_n \log f_\theta(Y| B^\top X) 
\mbox{ and }
a_{\theta}(B)=
\mathbb P \log f_\theta(Y| B^\top X).
\]
Then
\[
M_n(B)=\sup_{\theta\in\Theta_n} a_{n,\theta}(B)
\mbox{ and }
M_n^{\mathrm{pop}}(B)=\sup_{\theta\in\Theta_n} a_{\theta}(B).
\]
Using the elementary inequality
\[
\left|
\sup_{\theta\in\Theta_n} a_{n,\theta}(B)
-
\sup_{\theta\in\Theta_n} a_{\theta}(B)
\right|
\le
\sup_{\theta\in\Theta_n}
\left|
a_{n,\theta}(B)-a_{\theta}(B)
\right|,
\]
we obtain
\[
\sup_{B\in\mathcal V_{p,d}}
\left|
M_n(B)-M_n^{\mathrm{pop}}(B)
\right|
\le
\sup_{\substack{B\in\mathcal V_{p,d}\\ \theta\in\Theta_n}}
\left|
\mathbb P_n \log f_\theta(Y| B^\top X)
-
\mathbb P\log f_\theta(Y| B^\top X)
\right|.
\]
The right-hand side converges to zero in probability by the assumed uniform law of large numbers. Hence we get the desired result
$\sup_{B\in\mathcal V_{p,d}}
\left|
M_n(B)-M_n^{\mathrm{pop}}(B)
\right|
\overset{p}{\longrightarrow}0$.
\qed

We remark that the ULLN assumption in Lemma~\ref{lem:ulln_sieve} is satisfied by conditional rational-quadratic spline flow classes under standard sieve restrictions, including uniformly bounded network weights, spline derivatives, knot locations, and tail behavior, together with suitably controlled growth of model complexity with $n$. See, for example, \citet{kobyzev2020normalizing} and \citet{irons2022triangular}. 
We further note that $M_n(B)$ in (\ref{eq:sieve1}) coincides with $M_n(B)$ in (\ref{eq:sample profile}) of the main text, while $M_n^{\mathrm{pop}}(B)$ in (\ref{eq:sieve2}) corresponds to the population profile criterion in (\ref{eq:population profile}).

\begin{lemma}[Uniform approximation of the population target]
\label{lem:bridge}
For the  information index $\mathcal{J}(B)={\rm E}\left\{ \log{f(Y | B^\top X)} \right\}$ 
defined in (\ref{eq:mi}) and
the population profile log-likelihood  $M_n^{\mathrm{pop}}(B)$ in (\ref{eq:sieve2}),
assume
 that the sieve approximation error satisfies
\begin{equation}
\label{eq:sieve3}
\sup_{B \in \mathcal V_{p,d}}
\left|
M_n^{\mathrm{pop}}(B) - \mathcal J(B)
\right|
\le \varepsilon_n,
\mbox{ for some }
\varepsilon_n \to 0.
\end{equation}
In addition, assume the ULLN assumption in Lemma~\ref{lem:ulln_sieve} is satisfied. Then
\[
\sup_{B \in \mathcal V_{p,d}}
\left|
M_n(B) - \mathcal J(B)
\right|
\overset{p}{\longrightarrow} 0.
\]

\end{lemma}

\noindent {\bf Proof}:
By the triangle inequality,
\[
\begin{aligned}
\sup_{B\in\mathcal V_{p,d}}
\left|M_n(B)-\mathcal J(B)\right|
&\le
\sup_{B\in\mathcal V_{p,d}}
\left|M_n(B)-M_n^{\mathrm{pop}}(B)\right| +
\sup_{B\in\mathcal V_{p,d}}
\left|M_n^{\mathrm{pop}}(B)-\mathcal J(B)\right|.
\end{aligned}
\]
The first term converges to zero in probability by Lemma~\ref{lem:ulln_sieve}, while the second
term is bounded by $\varepsilon_n=o(1)$. Therefore we get the desired result
$\sup_{B\in\mathcal V_{p,d}}
\left|M_n(B)-\mathcal J(B)\right|
\overset{p}{\longrightarrow}0$.
\qed

We remark that the approximation in (\ref{eq:rich class1}) of the main text is formally characterized by the uniform sieve approximation condition (\ref{eq:sieve3}). 
This condition is plausible for conditional rational-quadratic spline flow classes, since monotone rational-quadratic splines provide flexible one-dimensional density transforms, and neural networks can approximate the dependence of spline parameters on the conditioning variable $B^\top X$. 
However, this condition is not automatic. It requires standard regularity conditions on the true conditional density, as well as compactness or appropriate tail control, bounded spline parameters, and sufficiently rich yet controlled sieve growth. 
We therefore impose it as a sieve approximation assumption and leave its formal verification for future work.

For \(A,B\in\mathcal V_{p,d}\), define
$$\Delta(A,B)=\|AA^\top-BB^\top\|_{\rm F},$$
where $\|\,\cdot\,\|_{\rm F}$ denotes the Frobenius norm. 
This metric depends only on the column spaces of \(A\) and \(B\), and therefore quantifies the distance between the subspaces \(\operatorname{span}(A)\) and \(\operatorname{span}(B)\).

\begin{theorem}[Consistency of the estimated central subspace]
\label{thm:subspace_consistency}
Suppose $\mathcal{S}_{Y|X}$ exists and has structural dimension $d$. 
Assume that the uniform law of large numbers (ULLN) condition in Lemma~\ref{lem:ulln_sieve} and  the uniform sieve approximation
condition in Lemma~\ref{lem:bridge} hold. 
Then 
$$\Delta(\hat B,B_0)
\overset{p}{\longrightarrow}0,$$
where $\hat B
=
\arg\max_{B\in\mathcal V_{p,d}} M_n(B)$ is defined in (\ref{eq:sample}) and  $B_0=\arg\max_{B \in \mathcal V_{p,d}} \mathcal{J}(B)$ is defined in (\ref{eq:population}). 
\end{theorem}

\noindent {\bf Proof}:
By Lemma~\ref{lem:bridge},
\[
\sup_{B\in\mathcal V_{p,d}}
\left|M_n(B)-\mathcal J(B)\right|
\overset{p}{\longrightarrow}0.
\]

Part 4 of Lemma \ref{lem:order} suggests that \(\mathcal J(B)\) has a unique maximizer among $B\in\mathcal V_{p,d}$, which is denoted as $B_0$. 
For any \(\epsilon>0\), define
\[
\mathcal A_\epsilon
=
\{B\in\mathcal V_{p,d}:\Delta(B,B_0)\ge \epsilon\}.
\]
The set \(\mathcal A_\epsilon\) is compact and does not contain $B_0$. By the subspace-level uniqueness of the maximizer,
\[
\sup_{B\in\mathcal A_\epsilon}\mathcal J(B)
<
\mathcal J(B_0).
\]
Let
\[
\eta_\epsilon
=
\mathcal J(B_0)
-
\sup_{B\in\mathcal A_\epsilon}\mathcal J(B)
>0.
\]

On the event
\[
\sup_{B\in\mathcal V_{p,d}}
\left|M_n(B)-\mathcal J(B)\right|
<
\eta_\epsilon/3,
\]
we have
\[
M_n(B_0)
>
\mathcal J(B_0)-\eta_\epsilon/3
\]
and, for every \(B\in\mathcal A_\epsilon\),
\[
M_n(B)
<
\mathcal J(B)+\eta_\epsilon/3
\le
\mathcal J(B_0)-\eta_\epsilon+\eta_\epsilon/3
=
\mathcal J(B_0)-2\eta_\epsilon/3.
\]
Therefore,
\[
M_n(B_0)
>
\sup_{B\in\mathcal A_\epsilon}M_n(B).
\]
Thus any maximizer \(\hat B\arg\max_{B\in\mathcal V_{p,d}}M_n(B)\) cannot belong to
\(\mathcal A_\epsilon\). Hence,
\[
\Pr\{\Delta(\hat B,B_0)\ge \epsilon\}
\le
\Pr\left\{
\sup_{B\in\mathcal V_{p,d}}
|M_n(B)-\mathcal J(B)|
\ge
\eta_\epsilon/3
\right\}
\longrightarrow 0.
\]
Since this holds for every \(\epsilon>0\), we conclude that
$\Delta(\hat B,B_0)
\overset{p}{\longrightarrow}0$.
\qed

\section{Population-level comparisons with the benchmark methods}
\label{sec:population comparison}

BENN-SDR \citep{tang2026belted} considers a parametric family of functions
\[
\mathcal{G}
=
\{ g(\cdot,t) : t \in \mathcal{C} \},
\]
where $\mathcal{C}$ is a subset of  $\mathbb{R}$. For example, $\mathcal{G}_m
=
\{ g(y,t)=y^t : t =1,2,\ldots,m \}$ defines the first $m$ moment functions. With the ensemble $\mathcal{G}$, the population-level linear BENN-SDR solves
\[
\min_{B,h}
\int_{\mathcal{C}}
{\rm E}\left\{
g(Y,t)-h_t(B^\top X)
\right\}^2
\, d\mu(t),
\]
where $h_t(\cdot)$ belongs to a sufficiently rich function class, and $\mu$ is a finite measure on ${\mathcal{C}}$.  At the sample level, BENN-SDR models $h_t(\cdot)$ with a standard feedforward multilayer perceptron (MLP). In the special case where $\mathcal{G}=\mathcal{G}_m$ with $m=1$, BENN-SDR reduces to NN-SDR of
\cite{kapla2022fusing}. 

When we choose $\mathcal{G}_m$ with $m>1$,  
BENN employs a feedforward MLP with a low-dimensional bottleneck layer and an $m$-dimensional output layer corresponding to the ensemble functions,
so that  $\{h_t(\cdot)\}_{t=1}^m$
can be learned simultaneously. In this case, BENN targets the moment-based dimension reduction problem
$${\rm E}\{(Y,Y^2,\ldots,Y^m)\mid X \}={\rm E}\{(Y,Y^2,\ldots,Y^m)\mid  B^\top X \},$$
which was first studied in \cite{yin2002dimension}. 

Although both linear BENN and FlowSDR target the same sufficient predictor $B^\top X$, they differ in the way the conditional distribution is learned. BENN recovers the central subspace through a least-squares regression of an ensemble of transformed responses on $B^\top X$, whereas FlowSDR directly maximizes a conditional likelihood induced by a flexible conditional normalizing flow family.

The population-level  MAVE \citep{xia2002mave} solves $$\min_{B} {\rm E}\left\{
\bigl(Y - E(Y \mid B^\top X)\bigr)^2
\right\},$$ which is equivalent to 
\[
\min_{B} {\rm E}\big\{\operatorname{Var}(Y \mid B^\top X)\big\}.
\]
Under the neural Gaussian SDR  working model (\ref{eq:ng model}), we assume
\begin{equation*}
Y \mid B^\top X = w
\sim
\mathcal{N}\big( \mu(w), \sigma^2(w) \big).
\end{equation*}
Ignoring the constant
$-\log(2\pi)/2$,
the corresponding population log-likelihood becomes
\begin{equation*}
\ell^{\rm G}(B,\mu,\sigma^2)
=
-\frac12
{\rm E}\left\{
\log \sigma^2(B^\top X)
+
\frac{(Y-\mu(B^\top X))^2}
{\sigma^2(B^\top X)}
\right\}.
\end{equation*}
For a fixed \(B\), the population profile log-likelihood is
$$M^{\rm G}(B)=\sup_{\mu,\sigma^2}\ell^{\rm G}(B,\mu,\sigma^2).$$
The population-level  neural Gaussian SDR estimand is then defined as
\begin{equation}
\label{eq:gaussian_population}
{B}^{\rm G}
=
\arg\max_{ B \in\mathcal V_{p,d}} M^{\rm G}(B).
\end{equation}

The next proposition states that the population-level target of neural Gaussian SDR  is closely related to MAVE. In fact,  we can view it as a log conditional variance analogue of MAVE.

\begin{proposition}[Profile Gaussian likelihood]
\label{prop:profile_gaussian}
Assume that \(\operatorname{Var}(Y\mid B^\top X)>0\) almost surely for each
\(B\in\mathcal V_{p,d}\).
The population-level  neural Gaussian SDR estimand 
satisfies
 \[
{B}^{\rm G}
=
\arg\min_{ B \in\mathcal V_{p,d}}  {\rm E}  \big\{
\log\operatorname{Var}(Y \mid B^\top X)\big\}.
\]
\end{proposition}

\noindent
{\bf Proof.} From the law of iterated expectations, we have
$$\ell^{\rm G}(B,\mu,\sigma^2)
=
-\frac12 {\rm E}\left[
\left.{\rm E}\left\{
\log \sigma^2(B^\top X)
+
\frac{(Y-\mu(B^\top X))^2}
{\sigma^2(B^\top X)}
\right | B^\top X \right\} \right].$$
Conditioning on $B^\top X=w$ and
for fixed $\sigma^2(w)$, the term involving $\mu(w)$ is
\[
-{\rm E}\left
\{(Y-\mu(w))^2
\mid B^\top X=w
\right\}
/
\sigma^2(w),
\]
which is maximized at
\[
\mu_B(w)
=
{\rm E}(Y \mid B^\top X=w).
\]

Substituting this value gives the conditional contribution
\[
-\frac12
\left\{
\log \sigma^2(w)
+
\frac{
\operatorname{Var}(Y \mid B^\top X=w)
}{
\sigma^2(w)
}
\right\}.
\]

For fixed $w$, differentiating with respect to $\sigma^2(w)$
yields
\[
-\frac{1}{\sigma^2(w)}
+
\frac{
\operatorname{Var}(Y \mid B^\top X=w)
}{
\sigma^4(w)
}
=
0,
\]
and the maximizer is 
\[
\sigma_B^2(w)
=
\operatorname{Var}(Y \mid B^\top X=w).
\]

At this optimum, the ratio term equals one. Therefore
$$M^{\rm G}(B)=\sup_{\mu,\sigma^2}\ell^{\rm G}(B,\mu,\sigma^2)=-\frac12
{\rm E} \big\{
\log\operatorname{Var}(Y \mid B^\top X)+1\big\}. $$
The result then follows directly from the definition of ${B}^{\rm G}$ in (\ref{eq:gaussian_population}). \qed

\section{Implementation details}
\label{sec:details}

All experiments are carried out on a single NVIDIA RTX 5090 GPU. 

MAVE and csMAVE are implemented via the $R$ package ``${\rm MAVE}$''. 

NN-SDR and BENN-SDR are implemented following  \cite{tang2026belted} and \href{https://github.com/tyy20/BENN-codes}{BENN-codes GitHub repository}. 

The neural Gaussian SDR model uses a two-hidden-layer multilayer perceptron with architecture
\[
d \rightarrow 64 \xrightarrow{\tanh} 64\xrightarrow{\tanh} 2.
\]
The network maps the reduced predictor $B^\top X \in \mathbb{R}^d$ to the conditional mean and log-standard-deviation of $Y \mid B^\top X$.

For FlowSDR, the model have two parts:
\begin{enumerate}
\item 
\[
X \rightarrow W
\quad \text{with} \quad
W = B^\top X.
\]

\item
\[
(W,Y)\xrightarrow{\rm MLP} Z
\quad \text{with} \quad
Z\sim \mathcal{N}(0,1).
\]
\end{enumerate}

For the MLP in part 2, we use a two-hidden-layer multilayer perceptron with architecture
\[
(d+1)\rightarrow 64\xrightarrow{\tanh} 64\xrightarrow{\tanh}\text{spline parameters}. 
\]

For both neural Gaussian SDR and FlowSDR,
gradient descent is carried out jointly over the neural network parameters and the SDR projection matrix $B$ using the Adam optimizer. Separate learning rates are used for the two parameter groups: the flow network parameters use step size $10^{-3}$, while the projection matrix $B$ uses a larger step size $10^{-2}$. 

After each Adam update, the projection matrix $B$ is orthonormalized through QR decomposition to enforce the Stiefel manifold constraint $B^\top B=I_d$. Specifically, if $\widetilde B$ denotes the unconstrained gradient update, then its QR decomposition
$\widetilde B = QR$
is computed, and the orthonormal factor $Q$ is taken as the updated value of $B$. This operation can be viewed as a retraction step that maps the unconstrained gradient iterate back onto the Stiefel manifold
$\mathcal V_{p,d}$.

For the real data analysis, we summarize the protocol in Table~\ref{tab:utkface-protocol}. 
 We use a single stratified split consisting of 70\% training, 15\% validation, and 15\% test samples, where the stratification is based on age bins,
 and filenames are saved to ensure identical examples across all methods. 
The validation set is used for early stopping and model selection during neural network training, while the test set is reserved exclusively for 
final evaluation.

\begin{table}[H]
  \centering
    \caption{UTKFace data analysis protocol.}
  \label{tab:utkface-protocol}
  \begin{tabular}{p{0.27\linewidth}p{0.63\linewidth}}
    \toprule
    Component & Specification \\
    \midrule
    Dataset & UTKFace face images with continuous age labels \\
    Response & Age in years, parsed from filename \\
    Image preprocessing & Aligned and cropped images resized to $64 \times 64$ \\
    Feature extractor & HOG descriptors concatenated with RGB color histograms \\
    Feature dimension & $X \in \mathbb{R}^{1812}$ \\
    SDR index & $W = B^\top X$ with $B\in \mathbb{R}^{1812\times d} $ \\
    Candidate dimensions & $d \in \{2,4,6,8\}$ \\
    Split & 70\% train, 15\% validation, 15\% test \\
    Split control & Stratify by age bins\\
    Metrics & MAE, RMSE, and NLL  on the test data\\
    \bottomrule
  \end{tabular}
\end{table}

\section{Order determination}
\label{sec:order}

Estimating the structural dimension $d$ is known as order determination in the SDR literature. 
Sequential tests and BIC-type procedures are commonly used for order determination in classical SDR methods. 
However, order determination is substantially more challenging for neural-network-based SDR methods. To the best of our knowledge, there is currently no theoretically grounded procedure for estimating $d$ other than the approach proposed in \cite{yang2025golden}. 

Specifically, among $1\le k\le p$, the estimator of $d$ is defined as the minimizer of the criterion
$$
C(k)
=
{\rm MSE}_{\rm validation}(k)
+
k\times {\rm pen}(n_{\rm validation},n_{\rm train}),
$$
where ${\rm MSE}_{\rm validation}(k)$ denotes the validation MSE obtained under working dimension $k$, and 
${\rm pen}(n_{\rm validation},n_{\rm train})$ is a penalty term depending on both the validation sample size $n_{\rm validation}$ and the training sample size $n_{\rm train}$. 
\cite{yang2025golden} established the consistency of 
$\hat d=\arg\min_k C(k)$,
by deriving precise convergence rates for the corresponding neural-network-based SDR estimator of the central subspace. 

In contrast, while we establish consistency of FlowSDR, we do not derive convergence rates for the estimated central subspace. Consequently, we do not yet have a fully developed penalized-MSE framework for provably consistent order determination. Instead, we evaluate unpenalized MAE and RMSE for order determination. In particular, we use $n_{\rm train}=2000$, $n_{\rm validation}=1000$, and $n_{\rm test}=2000$. The validation data are used for tuning and selecting the fitted neural-network model within each fixed working dimension $k$, while the final criterion (MAE or RMSE) for order determination is evaluated on the test data set.

\begin{table}[H]
  \centering
  \scriptsize
  \setlength{\tabcolsep}{4pt}
\caption{Results of estimating $d=2$ in Model I with $100$ replications} 
  \label{tab:fv-order-cv-normal}
  \begin{tabular}{clcrrrc}
    \toprule
    $p$ & Method & Criterion & $\hat d=1$ & $\hat d=2$ & $\hat d=3$ & $\hat d\ge 2$ \\
    \midrule
    5 & neural Gaussian SDR & MAE & 0 & 64 & 36 & 100 \\
     & neural Gaussian SDR & RMSE & 0 & 58 & 42 & 100 \\
     & FlowSDR & MAE & 0 & 61 & 39 & 100 \\
     & FlowSDR & RMSE & 0 & 56 & 44 & 100 \\
    \addlinespace[2pt]
    10 & neural Gaussian SDR & MAE & 0 & 70 & 30 & 100 \\
     & neural Gaussian SDR & RMSE & 0 & 70 & 30 & 100 \\
     & FlowSDR & MAE & 0 & 68 & 32 & 100 \\
     & FlowSDR & RMSE & 0 & 63 & 37 & 100 \\
    \addlinespace[2pt]
    20 & neural Gaussian SDR & MAE & 0 & 60 & 40 & 100 \\
     & neural Gaussian SDR & RMSE & 0 & 60 & 40 & 100 \\
     & FlowSDR & MAE & 0 & 77 & 23 & 100 \\
     & FlowSDR & RMSE & 0 & 75 & 25 & 100 \\
    \bottomrule
  \end{tabular}
\end{table}

\begin{table}[H]
  \centering
  \scriptsize
  \setlength{\tabcolsep}{4pt}
\caption{Results of estimating $d=2$ in Model II with $100$ replications} 
  \label{tab:fv-order-cv-t3}
  \begin{tabular}{clcrrrc}
    \toprule
    $p$ & Method & Criterion & $\hat d=1$ & $\hat d=2$ & $\hat d=3$ & $\hat d\ge 2$ \\
    \midrule
    5 & neural Gaussian SDR & MAE & 0 & 51 & 49 & 100 \\
     & neural Gaussian SDR & RMSE & 0 & 48 & 52 & 100 \\
     & FlowSDR & MAE & 0 & 72 & 28 & 100 \\
     & FlowSDR & RMSE & 1 & 57 & 42 & 99 \\
    \addlinespace[2pt]
    10 & neural Gaussian SDR & MAE & 0 & 55 & 45 & 100 \\
     & neural Gaussian SDR & RMSE & 0 & 51 & 49 & 100 \\
     & FlowSDR & MAE & 0 & 75 & 25 & 100 \\
     & FlowSDR & RMSE & 2 & 60 & 38 & 98 \\
    \addlinespace[2pt]
    20 & neural Gaussian SDR & MAE & 0 & 56 & 44 & 100 \\
     & neural Gaussian SDR & RMSE & 0 & 50 & 50 & 100 \\
     & FlowSDR & MAE & 0 & 67 & 33 & 100 \\
     & FlowSDR & RMSE & 0 & 66 & 34 & 100 \\
    \bottomrule
  \end{tabular}
\end{table}

From Tables~\ref{tab:fv-order-cv-normal} and~\ref{tab:fv-order-cv-t3}, we see that unpenalized MAE and RMSE tend to overestimate $d$. This is consistent with the findings of \cite{yang2025golden}, where a penalized MSE criterion was shown to yield consistent estimation of $d$. We argue that moderate overestimation of $d$ is not a major concern. On one hand, part~5 of Lemma~\ref{lem:order} suggests that using $\hat d>d$ does not lead to loss of information. On the other hand, neural-network-based methods empirically appear less sensitive to overfitting than many classical statistical methods.

For Model II with heavy-tailed errors, we see from Table~\ref{tab:fv-order-cv-t3} that using RMSE may underestimate $d$, whereas MAE does not appear to suffer from this issue. This is a more serious concern, since the conditional distribution $f(Y| \hat B^\top X)$ can differ substantially from $f(Y| X)$ when $\hat B$ is estimated under $\hat d<d$. We expect that a penalized MSE criterion may encounter similar issues, as penalization tends to favor more parsimonious models. Developing a penalized MAE-type criterion that is robust to heavy-tailed distributions warrants further investigation.

\section{Conformal prediction intervals}
\label{sec:ci}

{\bf Conformal Neural Gaussian prediction intervals.} Let $\{(X_i, Y_i)\}_{i=1}^{n_{\mathrm{cal}}}$ denote an independent calibration set, and let $\hat{\mu}(\cdot)$ and $\hat{\sigma}(\cdot)$ be the fitted mean and standard deviation functions from the neural Gaussian SDR model, together with the estimated central subspace basis $\hat B$.
We define the heteroscedastic conformal scores as
\[
R_i
=
\frac{\big|Y_i - \hat{\mu}(\hat{B}^\top X_i)\big|}
{\hat{\sigma}(\hat{B}^\top X_i)},
\quad i = 1, \dots, n_{\mathrm{cal}}.
\]

Let
\[
\hat{q}_{1-\alpha}
=
\text{the } \left\lceil (n_{\mathrm{cal}} + 1)(1-\alpha) \right\rceil
\text{-th order statistic of } \{R_i\}_{i=1}^{n_{\mathrm{cal}}}.
\]

Then, for a new input $X_{\rm new}$, the conformal neural Gaussian prediction interval at level $1-\alpha$ is given by
\[
\mathcal{C}_{1-\alpha}(X_{\rm new})
=
\left[
\hat{\mu}(\hat{B}^\top X_{\rm new})
-
\hat{q}_{1-\alpha}\,\hat{\sigma}(\hat{B}^\top X_{\rm new}),
\;
\hat{\mu}(\hat{B}^\top X_{\rm new})
+
\hat{q}_{1-\alpha}\,\hat{\sigma}(\hat{B}^\top X_{\rm new})
\right].
\]

Under the assumption of exchangeability, this interval satisfies the finite-sample coverage guarantee
\[
\Pr\big( Y_{\rm{new}} \in \mathcal{C}_{1-\alpha}(X_{\rm{new}}) \big)
\ge 1 - \alpha.
\]

Compared with the Gaussian model-based interval
\[
\hat{\mu}(\hat{B}^\top X_{\rm new})
\pm
z_{1-\alpha/2}\,\hat{\sigma}(\hat{B}^\top X_{\rm new}),
\]
the conformal interval replaces the normal critical value with a data-driven quantile $\hat{q}_{1-\alpha}$, thereby correcting for model misspecification and estimation error.

{\bf Conformal FlowSDR prediction intervals.}
Let $\{(X_i, Y_i)\}_{i=1}^{n_{\mathrm{cal}}}$ denote an independent calibration set, and let $\hat{F}(\,\cdot \mid \hat{B}^\top x)$ denote the fitted conditional distribution function from the FlowSDR model, together with the estimated central subspace basis $\hat B$. Denote the corresponding conditional quantile function by
\[
\hat{Q}_\tau(\hat{B}^\top x)
=
\inf\{y : \hat{F}(y \mid \hat{B}^\top x) \ge \tau\}.
\]

For a target miscoverage level $\alpha \in (0,1)$, define the model-based central interval as
\[
\left[
\hat{Q}_{\alpha/2}(\hat{B}^\top x),
\;
\hat{Q}_{1-\alpha/2}(\hat{B}^\top x)
\right].
\]

To correct for model misspecification, we define conformal scores on the calibration set as
\[
R_i
=
\max\Big\{
\hat{Q}_{\alpha/2}(\hat{B}^\top X_i) - Y_i,
\;
Y_i - \hat{Q}_{1-\alpha/2}(\hat{B}^\top X_i),\; 0
\Big\},
\quad i = 1, \dots, n_{\mathrm{cal}}.
\]

Let
\[
\hat{q}_{1-\alpha}
=
\text{the } \left\lceil (n_{\mathrm{cal}} + 1)(1-\alpha) \right\rceil
\text{-th order statistic of } \{R_i\}_{i=1}^{n_{\mathrm{cal}}}.
\]

Then, for a new input $X_{\rm new}$, the conformal FlowSDR prediction interval is given by
\[
\mathcal{C}_{1-\alpha}(X_{\rm new})
=
\left[
\hat{Q}_{\alpha/2}(\hat{B}^\top X_{\rm new}) - \hat{q}_{1-\alpha},
\;
\hat{Q}_{1-\alpha/2}(\hat{B}^\top X_{\rm new}) + \hat{q}_{1-\alpha}
\right].
\]

Under exchangeability, this interval satisfies the finite-sample coverage guarantee
\[
\Pr\big( Y_{\rm new} \in \mathcal{C}_{1-\alpha}(X_{\rm new}) \big)
\ge 1 - \alpha.
\]

Compared with the neural Gaussian SDR interval, which assumes a parametric Gaussian form, FlowSDR directly models the full conditional distribution and captures non-Gaussian features such as skewness and multimodality. The conformal correction further ensures valid coverage under model misspecification and finite-sample effects.

\end{document}